\documentclass[reqno,a4paper,10pt]{amsart}
\usepackage[a4paper, centering]{geometry}
\usepackage[utf8]{inputenc} % allow utf-8 input
\usepackage[english]{babel}
\usepackage[T1]{fontenc}
\usepackage{lmodern}
\usepackage{textcomp} 
\usepackage{caption}
\usepackage{subcaption}
\usepackage{enumerate}
\usepackage{amssymb, amsmath}
\usepackage{mathrsfs,dsfont}
\usepackage{amscd}
\usepackage{bbm}
\usepackage{amsthm}
\usepackage{cite}
\usepackage{esint }
\usepackage[mathcal]{euscript}
\usepackage[active]{srcltx}
\usepackage{verbatim}
\usepackage[colorlinks,linkcolor={blue},citecolor={blue},urlcolor={black}]{hyperref}
\usepackage[]{changebar}
\usepackage{xcolor}
\usepackage{mathtools}
\usepackage{url}            % simple URL typesetting
\usepackage{booktabs}       % professional-quality tables
\usepackage{amsfonts}       % blackboard math symbols
\usepackage{nicefrac}       % compact symbols for 1/2, etc.
\usepackage{microtype}      % microtypography
\usepackage{lipsum}
\usepackage{fancyhdr}       % header
\usepackage{graphicx}       % graphics
\graphicspath{{media/}}     % organize your images and other figures under media/ folder
\usepackage{bookmark} %(fa vedere le formule)
\usepackage[a4paper, centering]{geometry}
\usepackage{epstopdf}
\usepackage{amscd}
\usepackage{color}
\usepackage{graphicx}

\usepackage{yfonts}
\usepackage{fullpage}
\usepackage{comment}
\usepackage[braket,qm]{qcircuit} % quantum circuits
\usepackage{ stmaryrd } % per un commando che chiamo lket dopo

\usepackage[many]{tcolorbox}
\newtcolorbox{mybox}{enhanced,colback=red!5!white, colframe=red!75!black, width=\textwidth,box align=center,halign=center,valign=center, center}
\usepackage[new]{old-arrows} % large arrows
\usepackage{mathabx} % updownarrows

\usepackage{tabularx}
\usepackage{multirow}
\usepackage{makecell}

\newtheorem{thm}{Theorem}[section]

\newtheorem*{thm*}{Theorem}
\newtheorem{cor}{Corollary}[section]

\newtheorem{lem}{Lemma}[section]

\newtheorem{prop}{Proposition}[section]

\newtheorem*{prop*}{Proposition}
\newtheorem{ass}{A}

\theoremstyle{definition}
\newtheorem{defn}{Definition}[section]

\theoremstyle{remark}

\numberwithin{equation}{section}

\def\N{{\mathbb N}}

\def\R{{\mathbb R}}

\def\W{{\mathcal W}}

\def\norma #1{\left\lVert #1 \right\rVert}

\def\P{{\mathscr{P}}}

\def\E{{\mathbb E}}

\def\H{{\mathscr H}}

\def\M{{\mathcal M}}
\def\L{{\mathcal L}}
\def\T{{\mathbb T}}

\def\X{{\mathscr X}}

\def\O{{\mathcal O}}

\newcommand{\hol}[1]{\"{#1}}
\def\D{{\mathcal D}}

\def\N{{\mathbb N}}

\def\R{{\mathbb R}}
\def\CC{{\mathbb C}}
\def\L{{\mathcal L}}

\def\H{{\mathcal H}}
\def\O{{\mathcal O}}
\def\E{{\mathbb E}}

\def\X{{\mathcal X}}

\def\D{{\mathcal D}}

\def\norma #1{\left\lVert #1 \right\rVert}

\def\P{{\mathbb{P}}}

\def\de{{\rm d}}

\def\ov #1{\overline{#1}}

\definecolor{viola}{rgb}{0.3,0,0.7}
\definecolor{ciclamino}{rgb}{0.5,0,0.5}
\definecolor{rosso}{rgb}{0.8,0,0}

\newcommand{\beq}{\begin{equation}}
\newcommand{\eeq}{\end{equation}}
\newcommand{\bal}{\begin{aligned}}
\newcommand{\eal}{\end{aligned}}
\newcommand{\ben}{\begin{enumerate}}
\newcommand{\beni} {\begin{enumerate}[(i)]}
\newcommand{\een}{\end{enumerate}}
\newcommand{\bit}{\begin{itemize}}
\newcommand{\eit}{\end{itemize}}
\newcommand{\beqw}{\begin{equation*}}
\newcommand{\eeqw}{\end{equation*}}
\newcommand{\bex}{\begin{example}}
\newcommand{\eex}{\end{example}}
\newcommand{\bre}{\begin{example}}
\newcommand{\ere}{\end{example}}
\newcommand{\bma}{\begin{bmatrix}}
\newcommand{\ema}{\end{bmatrix}}

\makeatletter
\@namedef{subjclassname@2020}{%
  \textup{2020} Mathematics Subject Classification}
\makeatother

\title[MFG]{Mean-field limit from general mixtures of experts to quantum neural networks}
\author[A.~Melchor Hernandez]{Anderson Melchor Hernandez}
\address[A.~Melchor Hernandez]{Dipartimento di Matematica, Via Zamboni, 33, 40126, Bologna (Italy)}

\author[D.~Pastorello]{Davide Pastorello}
\address[D.~Pastorello]{Dipartimento di Matematica,
Università di Bologna, Via Zamboni, 33, 40126, Bologna (Italy)}
\email{davide.pastorello3@unibo.it}

\author[G.~De Palma]{Giacomo De Palma}
\address[G.~De Palma]{Dipartimento di Matematica,
Università di Bologna, Via Zamboni, 33, 40126, Bologna (Italy)}
\email{giacomodepalma@unibo.it}

\email{anderson.melchor@unibo.it}
\date{\today}

\begin{document}

\begin{abstract}
In this work, we study the asymptotic behavior of Mixture of Experts (MoE) trained via gradient flow on supervised learning problems. Our main result establishes the propagation of chaos for a MoE as the number of experts diverges. We demonstrate that the corresponding empirical measure of their parameters is close to a probability measure that solves a nonlinear continuity equation, and we provide an explicit convergence rate that depends solely on the number of experts. We apply our results to a MoE generated by a quantum neural network.
\end{abstract}

\maketitle

\tableofcontents

\section{Introduction}

In recent years, there has been a growing interest in leveraging Artificial Intelligence methods for the analysis of large datasets \cite{berlyand2023,russell2016}. The umbrella term ``Artificial Intelligence'' encompasses numerous subfields, ranging from learning theory to the mathematical foundations of its development. At the heart of AI and machine learning is the detection of intricate patterns within massive amounts of data, an ability that enables systems to discover hidden relationships, generate insightful predictions, and draw conclusions that might otherwise elude human observation \cite{bishop2006}. %In the context of supervised learning, where a training set of labeled data is available, a mixture of experts (MoE) is a machine learning architecture that combines multiple specialized models (experts) to solve complex tasks. This approach enhances scalability, efficiency, and performance by leveraging the strengths of specialized models while maintaining flexibility through adaptive weighting \cite{jordan1994hierarchical}. In particular, each expert can be given by a quantum machine learning model as we argue in the paper. 
Among newly emerging AI disciplines, Quantum Machine Learning stands out for uniting classical machine learning techniques with quantum computing \cite{de2019primer,schuld2015,pastorello2023concise}. The key principle behind quantum machine learning is to exploit quantum algorithms and quantum-mechanical effects---such as superposition, entanglement, and quantum parallelism---to improve the performance of deep neural models \cite{biamonte2017}. One of the most prominent QML algorithms is constituted by quantum neural networks, which constitute the quantum analog of deep neural networks. The output of a quantum neural network is the expectation value of a quantum observable measured on the state generated by a parametric quantum circuit. Such circuit is made by parametrized one-qubit and two-qubit gates \cite{girardi2024,schuld2018}, whose parameters encode both the input data and the trainable components of the model. Typically, these parameters are optimized via gradient-based methods to minimize a cost function and boost the ability of the circuit to process and analyze data \cite{schuld2021effect}.

In this paper, we investigate a general parametric model which may be implemented by a quantum neural network. Consider a finite set $\mathcal{X}$ of possible inputs (\emph{e.g.}, images embedded in $\mathbb{R}^d$) and let $\Theta$ represent the vector of circuit parameters. A parametric model defines a function $x \mapsto f(\Theta, x)$. Suppose we have a training set $\{(x^{(i)},\,y^{(i)}) : i=1,\dots,n\}$, where $x^{(i)}\in\mathcal{X}$ (\emph{e.g.}, dog or cat images) and $y^{(i)}\in\mathbb{R}$ (\emph{e.g.}, $y^{(i)}=1$ for a dog, $y^{(i)}=-1$ for a cat). The objective of supervised learning is to find parameters $\Theta$ such that $f(\Theta, x)$ closely matches the labels $y^{(i)}$. A common approach is to minimize the empirical quadratic loss
\begin{equation}
\mathcal{L}(\Theta) = \frac{1}{2}\sum_{i=1}^n \left(f(\Theta,x^{(i)}) - y^{(i)}\right)^2
\end{equation}
using gradient-based optimization. For simplicity, we analyze the continuous-time gradient flow instead of discrete gradient descent.

Recent work has sought to explore quantum neural networks in light of their potential to couple quantum computational power with the representational efficiency of deep learning \cite{lloyd2020quantum}. Indeed, a remarkable quantum speed-up can be achieved by a non-linear encoding of the data into a quantum feature space and using a linear classifier in a high-dimensional Hilbert space \cite{havlivcek2019,liu2021}. Despite such promising results, open questions remain regarding how to identify and optimize circuit parameters effectively \cite{cinelli2021var}.

Ref. \cite{girardi2024} considers quantum neural networks trained on supervised learning tasks, where the objective function is defined as the expected value of the sum of single-qubit observables across all qubits normalized such that its variance at initialization is $\Theta(1)$. The authors characterize the training dynamics of such quantum neural networks and prove their trainability in the limit of infinite width in any regime where the depth is allowed to grow with the number of qubits (denoted by $m$), as long as barren plateaus do not arise.
More precisely, Ref. \cite{girardi2024} proves that the probability distribution of the function generated by the trained network converges in distribution to a Gaussian process whose mean and covariance can be computed analytically \cite[Theorem 4.15]{girardi2024}.
Ref. \cite{hernandez2024} has subsequently provided quantitative bounds to such convergence in terms of the Wasserstein distance of order $1$.

In this paper, we adopt an alternative approach to the training of quantum neural networks: the {\em mean-field limit}, which has been extensively studied in the setting of classical deep neural networks \cite{sirignano2021meanfieldanalysisdeep,nguyen2019meanfieldlimitlearning,nguyen2023rigorous,mei2019meanfieldtheorytwolayersneural,lu2020meanfieldanalysisdeepresnet}.
%, where the generated function is normalized such that its absolute value is uniformly upper bounded.
%In the setting of classical deep neural networks, such mean-field limit has been extensively studied.
Such a limit approximates the empirical distribution of the neurons of the network with a smooth probability distribution. In this viewpoint, each neuron can be viewed as a particle that evolves under a suitable gradient flow in parameter space, and as the width tends to infinity, one can derive a limiting partial differential equation description of the training dynamics.
Specifically, consider a feedforward neural network with multiple layers. When the number of neurons per layer grows to infinity (while scaling the weights and biases appropriately), the empirical distribution of the neurons in each layer converges to a smooth measure. The function generated by the network then becomes governed by this measure, making it possible to track the evolution of the measure itself via mean-field gradient-flow equations. This approach provides a rigorous mathematical framework for understanding why very large networks can often escape poor local minima, achieve small training error, and even generalize well \cite{mei2019meanfieldtheorytwolayersneural,Rotskoff_2022,sirignano2021meanfieldanalysisdeep,araújo2019meanfieldlimitcertaindeep}.

Here, we explore whether a similar mean-field viewpoint can be adopted for quantum neural networks. We begin by adopting a more general point of view, considering a model function called {\em mixture of experts} (MoE)  given by the average of simpler identical parametric models \cite{eigen2013learning,jacobs1991adaptive,rasmussen2001infinite}, and we  study its behavior as the number of experts increases.
In classical machine learning, mixture of experts have been extensively used in several contexts,
including transformer-based large language models (LLMs) \cite{cai2024surveymixtureexperts,eigen2013learning,jacobs1991adaptive,rasmussen2001infinite}. 
We consider a uniform mixture of $N$ identical experts: 
\begin{equation}
    F(\Theta,x)\coloneqq\frac{1}{N}\sum_{i=1}^{N}f(\theta^{i},x),
\end{equation}
where $x$ is the input, $\theta^i$ are the parameters of the $i$-th expert, and $f$ is the model function of a single expert.

In the quantum case, each expert is a parametric quantum circuit with model function $f$ defined in \eqref{model1}. Then, $F$ turns out to be a hybrid model given by a classical mixture of quantum experts, which is hard to simulate classically if a significant fraction of the experts is hard to simulate.

In this spirit, we analyze the mean field phenomenon through the so-called propagation of chaos \cite{sznitman1991topics}. The propagation of chaos is a phenomenon observed in large systems of interacting particles, such as gases, where individual particles behave almost independently as the number of particles becomes infinitely large. This principle is central in statistical mechanics and kinetic theory, providing a bridge between microscopic dynamics and macroscopic laws \cite{sznitman1991topics}. We then apply this approach to analyze the collective behavior of a MoE, where the parameters $\Theta^{N}\coloneqq(\theta^{1},\ldots \theta^{N})$ are trained by gradient flow, resulting in an updated collection of parameter $\Theta_{t}^{N}\coloneqq(\theta_{t}^{1},\ldots, \theta_{t}^{N})$ that define a mixture of experts.
This, in turn, yields an updated model function $F(\Theta_{t},x)$.
The main idea of propagation of chaos is to compare the dynamics of $\Theta_{t}$, with the dynamics of a family of parameters $(\ov{\theta_{t}}^{1},\ldots, \ov{\theta_{t}}^{N})$ that are all independent, and then to study the proximity between their probability laws. In order to capture this proximity, we use the so-called Wasserstein distance \cite{kantorovich1960mathematical,MR2459454}. The Wasserstein distance of order $1\leq p<+\infty$, denoted as $\W_{p}$, is a metric employed to measure the distance between probability distributions on a metric space, capturing the geometry of the space in which they reside \cite{MR2459454}. Specifically, such distance is given by the $p$-th root of the minimum cost of transporting mass from one distribution to another, where the cost is determined by the $p$-th power of the distance between points of the ambient space \cite{kantorovich1960mathematical}.

\subsection{Our results}
In this paper, we prove the propagation of chaos for a sequence of $N\in\N$ experts whose dynamics follow the gradient flow equation associated with the minimization of the quadratic cost of a supervised learning problem. We show that, at each fixed time $t>0$, the empirical measure associated with these experts converges to a probability measure solving a nonlinear continuity equation. Without delving into all the details, we establish the following result (see \autoref{ourmainthm}, and \autoref{mfieldq} for a formal statement).

\begin{thm}\label{thm:informal}
Consider the MoE $F(\Theta,x)$ induced by the set of $N$ identical experts $\{f(\theta^{i},x):i=1,\ldots,N\}$ where $x$ represents a generic input, $\Theta\coloneqq(\theta^{1},\ldots \theta^{N})$ is the vector of parameters supported on the Torus $\mathbb{T}^{d}$ of dimension $d\in\N$ with period $2\pi$, and $f$ is a generic expert satisfying suitable regularity assumption on the variable $\theta$. Let each component of $\Theta$ be initialized by sampling it from the uniform distribution, and let then $\Theta$ be trained via gradient flow:
\begin{align}
\begin{aligned}
&\frac{\de \L(\Theta_{t}^{N})}{\de t}=-N\nabla_{\Theta}\L(\Theta_{t}^{N}),\\
&\L(\Theta_{t}^{N})\coloneqq \frac{1}{2}\sum_{i=1}^n \left(F(\Theta_{t}^{N},x^{(i)}) - y^{(i)}\right)^2.
\end{aligned}
\end{align}
Then, there exists a positive constant $C>0$ independent of $N$, and depending on $t$, such that

\begin{align}\label{bound1}
   \E\W_{2}(\mu_{\Theta_{t}^{N}},\mu_{t})\leq C\left(N^{-\frac{2}{d}}+N^{-\frac{1}{2}}\right),
\end{align}
where $\mu_{t}$ is the unique solution to the following nonlinear continuity equation:

\begin{align}
&\frac{\de \mu_{t}(\theta)}{\de t}=-\nabla_{\theta}\cdot\left(b(\theta,\mu_{t})\mu_{t}\right),\\
&b(\theta,\mu_{t})\coloneqq\sum_{j=1}^{n}\nabla_{\theta}f(\theta,x_{j})\left(y_{j}-\E_{\ov{\alpha}\sim \mu_{t}}\left[f(\ov{\alpha},x_{j})\right]\right).
\end{align}
Here, $\W_{2}$ denotes the Wasserstein distance of order $2$, and $\mu_{\Theta_{t}^{N}}$ is the empirical measure of the vector $\Theta_{t}^{N}$.
\end{thm}
%In particular, when the expert $f$ is defined through a quantum neural network, we find that the limit of the empirical measure of the trained parameters defining a MoE can be determined as the solution of a continuity equation \cite{ambrosio2008gradient}, where we have also provided the velocity of convergence to this limit.
Our bound in \eqref{bound1} explicitly depends on the dimension $d$. Notice that the right-hand side tends to zero as $N\rightarrow+\infty$.
Let us also note that our result does not hold in the limit $t\to\infty$, as our bound diverges for $t\to\infty$. To the best of our knowledge, the validity of the mean-field limit for infinite training time remains an open question even for classical deep neural networks \cite{araújo2019meanfieldlimitcertaindeep,sirignano2021meanfieldanalysisdeep,Rotskoff_2022}.

We then apply \autoref{thm:informal} to the setting where each expert is a parametric quantum circuit.
In the previous work \cite{hernandez2024}, the considered quantum neural network is a single parametric circuit. A proof of convergence of the probability distribution of the generated function to a Gaussian process is provided in the limit of infinite width, that is, in the limit of infinitely many qubits. In the present work, we consider a uniform mixture of experts givevn by fixed parametric quantum circuits, and we provide the proof of convergence of $\mu_{\Theta_{t}^{N}}$ to $\mu_t$ in the limit of infinitely many experts.   
Differently from the previous works \cite{girardi2024,hernandez2024}, where the asymptotic behavior of the quantum neural network is studied in the regime where the variance of the generated function at initialization is $\Theta(1)$, we study the regime where the generated function is uniformly bounded.
Such regime does not exhibit the lazy training of \cite{girardi2024,hernandez2024} that hinders representation learning.

This work is organized as follows. In \autoref{sec:prelim}, we set the notation of the paper, and we review the concepts of Wasserstein distance and propagation of chaos.
In \autoref{sec:generalmean}, we provide a general statement ensuring the existence of solutions for McKean-type equations where the drift term is determined by a general neural function satisfying appropriate growth conditions. Furthermore, in this section, we prove our main result, \autoref{ourmainthm}.  In \autoref{sub:quantum}, we apply our results to quantum neural networks. Finally, in \autoref{sec:concl}, we present some concluding remarks and discuss potential directions for future research.

\section{Preliminaries and notation}\label{sec:prelim}

Let us start by introducing the notation of the present work.
\subsection{Training data}\label{sub:trdata}
Let $\X$ be the feature space, \emph{i.e.}, the set of all the possible inputs, and we let $\mathbb{R}$ be the output space.
Let
\begin{equation}
    \D \coloneqq\left\{(x^{(i)},y^{(i)}):i=1,\ldots,n\right\}\subset \X\times\mathbb{R}
\end{equation}
be the training set. We set $n=\vert \D\vert$ to be the cardinality of $\D$.

\subsection{The Wasserstein distance of order 2}
In this subsection, we introduce the distance that we employ to quantify the closeness between the empirical measures associated to a vector of trained parameters and its weak limit. The Wasserstein distance of order $1\leq p<+\infty$ is a metric used to measure the distance between two probability distributions on a metric space, capturing not only differences in their values but also the geometry of the space in which they reside \cite{MR2459454,kantorovich1960mathematical}. Specifically, it is given by the $p$-th root of the minimum cost of transporting mass from one distribution to another, where the cost is determined by the $p$-th power of the distance between points of the ambient space. It admits applications in several areas of mathematics \cite{MR2459454} and in the general scenario of transporting resources in the cheapest way \cite{evans2012phylogenetic, rachev2011probability}. As a distance among probability distributions, the Wasserstein distance finds natural applications in statistics \cite{santambrogio2015optimal,panaretos2020invitation} and machine learning \cite{frogner2015learning, MAL-073, cheng2020wasserstein, arjovsky2017wasserstein} and it is also be extended to the quantum realm and considered in the context of quantum machine learning \cite{de2021quantum,kiani2022}.
Let $(\mathbb{T}^{d},\norma{\cdot}_{1})$ be the $d$-dimensional Torus with period $2\pi$ endowed with the $1$-norm $\norma{\cdot}_{1}$. In what follows, we set $\M(\mathbb{T}^{d})$ as the set of all probability measures over $\mathbb{T}^{d}$, and we endow it with $p$-Wasserstein distance defined as

\begin{align}
\W_{p}^{p}(\mu,\nu)\coloneqq\inf_{\pi\in \Gamma(\mu,\nu)}\displaystyle\int_{X\times X}(d_{X}(x_{1},x_{2}))^{p}\de \pi(x_{1},x_{2}),   
\end{align}
where $\Gamma(\mu,\nu)$ denotes the set of all possible joint probability measures having marginals $\mu$ and $\nu$. In what follows, we focus on $p=2$ as it relates to the concept of propagation of chaos, which will be discussed later.
\subsection{Particle systems and propagation of chaos}
In this section, we recall the notion of propagation of chaos as introduced by McKean and later developed by Sznitman to study the asymptotic behavior of large systems of particles \cite{sznitman1991topics,graham1992mckean}. Propagation of chaos refers to a phenomenon in systems of interacting particles where, as the number of particles grows, the behavior of any finite subset of particles becomes increasingly independent and identically distributed. This concept is particularly significant in the context of mean-field interactions, offering a rigorous connection between microscopic dynamics and macroscopic statistical behavior. It demonstrates how the collective evolution of a large system can often be effectively approximated by a limiting equation, such as the Vlasov or McKean-Vlasov equation, which governs the distribution of a single particle \cite{graham1996asymptotic}.

Let be given $N\in\N$, and suppose that 
 
\begin{align}\label{campi1}
b:\T^{d}\times \M(\T^{d})\rightarrow \T^{d},
\end{align}
is globally Lipschitz as in \eqref{Lipschitz}. Consider the $N$-particle system $\Theta_{t}^{N}\coloneqq (\theta_{t}^{1},\ldots \theta_{t}^{N})$ whose evolution is given by

\begin{align}\label{particleN}
\de \theta_{t}^{i}=b(\theta_{t}^{i},\mu_{\Theta_{t}^{N}})\de t, \hskip 0,2cm i\in\{1,\ldots,N\},
\end{align}
where

\begin{align}
\mu_{\Theta_{t}^{N}}\coloneqq \frac{1}{N}\sum_{i=1}^{N}\delta_{\theta_{t}^{i}}.
\end{align}
In what follows, we will present the seminal result stating that  as $N\rightarrow +\infty$, $\mu_{\Theta_{t}^{N}}$ converges to the unique solution $\mu_t$ of the nonlinear continuity equation

\begin{align}\label{FKP}
\frac{\de \mu_{t}(\theta)}{\de t}=-\nabla_{\theta}\cdot\left(b(\theta,\mu_{t})\mu_{t}\right),
\end{align}
and $\mu_{t}$ is the probability law of the so-called McKean process $(\ov{\theta}_{t})_{t\geq 0}$ which solves the following nonlinear differential equation 

\begin{align}\label{MKV}
\de \ov{\theta}_{t}=b(\ov{\theta}_{t},\mu_{t})\de t,
\end{align}
where $\mu_{t}={\mathrm Law}(\ov{\theta}_{t})$.

\begin{thm}\label{wellpod}
Suppose that $b$ is globally Lipschitz: there exists $C>0$ such that for all $x,y\in \T^{d}$ and for all $\mu,\nu\in \M(\T^{d})$ it holds that:

\begin{align}\label{Lipschitz}
\norma{b(x,\mu)-b(y,\nu)}_{1}\leq C\left(\norma{x-y}_{1}+ \W_{2}(\mu,\nu)\right).
\end{align}
Then for any $T>0$  the stochastic differential equation \eqref{MKV} has a unique strong solution on $[0,T]$, and consequently, its law is the unique weak solution to the continuity equation \eqref{FKP}.
\end{thm}
\begin{proof}
This is a consequence of \cite[Proposition 1]{Chaintron_2022a}. See also \cite[Theorem 2.1]{graham1992mckean}.
\end{proof}
The argument used to prove this Proposition is the classical Picard Iteration which can be used to prove the existence and uniqueness of a solution to the system \eqref{particleN}.
\begin{cor}\label{wellpodsyst}
Assume the same hypotheses as in \autoref{wellpod}.  Then, for any $T>0$ and any $i\in\{1,\ldots,N\}$, the system of stochastic differential equations \eqref{particleN} has a unique strong solution.
\end{cor}
\begin{proof}
This result can be derived as a particular case of the results proved in \cite{erny2022wellposedness}.
\end{proof}
In what follows, we recall the notion of propagation of chaos by coupling trajectories.

\begin{defn}\label{propchaos}
Let $T\in (0,+\infty]$, and $1\leq p<+\infty$. Propagation of chaos holds when for all $N\in\N$ there exist

\begin{itemize}
\item[$\bullet$] a system of particles $(\Theta_{t}^{N})_{t}=(\theta_{t}^{1},\ldots, \theta_{t}^{N})_{t}$ with law $\mu_{t}^{N}\in \M((\T^{d})^{N})$ at time $t\leq T$;
\item[$\bullet$] a system of independent stochastic processes $(\ov{\Theta}_{t}^{N})_{t}=(\ov{\theta}_{t}^{1},\ldots, \ov{\theta}_{t}^{N})$ with law $\mu_{t}^{\otimes N}\in \M((\T^{d})^{N})$ at time $t\leq T$, such that $\theta_{0}^{i}=\ov{\theta}_{0}^{i}$ $\P$-a.s. for $i=1,\ldots,N$;
\item[$\bullet$] a number $\varepsilon(N,T)$ such that $\varepsilon(N,T)\rightarrow 0$ as $N\rightarrow +\infty$,
\end{itemize}
such that (pathwise case)

\begin{align}
\frac{1}{N}\sum_{i=1}^{N}\E\left[\sup_{t\leq T}\norma{\theta_{t}^{i}-\ov{\theta}_{t}^{i}}_{1}^{p}\right]\leq \varepsilon(N,T),
\end{align}
or (pointwise case)
\begin{align}
\frac{1}{N}\sum_{i=1}^{N}\sup_{t\leq T}\E\left[\norma{\theta_{t}^{i}-\ov{\theta}_{t}^{i}}_{1}^{p}\right]\leq \varepsilon(N,T).
\end{align}
\end{defn}
Next, we prove propagation of chaos for a sequence of parameters $\Theta_{t}^{N}$ defined through the differential equation \eqref{particleN}. Furthermore, we prove the weak convergence in the sense of the probability measures of the initial measure $\mu_{\Theta_{0}^{N}}$ \cite{Billingsley-Convergence}.
\begin{thm}\label{thmpropchaos}
Let us assume the same hypotheses as in \autoref{wellpod}. Let $d>4$, and let us set $t\in [0,T]$, and $\Theta_{0}^{N}$ be composed by independent and identically distributed random variables supported in $\T^d$. Then the weak-limit in the sense of probability measures of $\mu_{\Theta_{0}^{N}}$, and denoted as

\begin{align}
\mu_{0}\coloneqq \lim_{N\rightarrow +\infty}\mu_{\Theta_{0}^{N}},
\end{align}
exists, and $\mu_{0}\in \M(\T^{d})$. Furthermore, there exists a sequence $(\ov{\Theta}_{t}^{N})_{t}=(\ov{\theta}_{t}^{1},\ldots, \ov{\theta}_{t}^{N})$ of independent, and identically distributed random variables valued in $\T^{d}$ such that the propagation of chaos in the sense of Definition \ref{propchaos} holds true with $p=2$. Additionally, the convergence rate $\varepsilon(N,T)$ is given by 
\begin{align}
\varepsilon(N,T)=C_{1}(b,T)\alpha_{d}(N),
\end{align}
where $C_{1}(b,T)$ is a positive constant depending only on $b,T$, and $\alpha_{d}(N)$ is given by
\begin{align}
\alpha_{d}(N)\coloneqq N^{-\frac{2}{d}}+N^{-\frac{1}{2}}.
\end{align}
\end{thm}
\begin{proof}
By assumption all the components of $\Theta_{0}^{N}$ are independent and identically distributed. Then for any measurable bounded function $g:\T^{d}\rightarrow \R$, we have by the law of large numbers that

\begin{align}
\int_{\T^{d}}g(\theta)\de\mu_{\Theta_{0}^{N}}(\theta)\rightarrow \E(g(\theta_{0}^{1})) \hskip 0,2cm \text{as $N\rightarrow+\infty$,}
\end{align}
so that $\mu_{0}={\rm law}(\theta_{0}^{1})$. Since $\theta_{0}^{1}$ is supported in $\T^{d}$, then $\mu_{0}\in \M(\T^d)$. Now, by following \cite[Theorem 3.20]{Chaintron_2022b}, we may consider as system of independent, and identically distributed particles the ones following \eqref{MKV} with the same initial condition $\Theta_{0}^{N}$ as the ones of \eqref{particleN}. So that, we have

\begin{align}
 \norma{\theta_{t}^{i}-\ov{\theta}_{t}^{i}}_{1}^{2}&\leq 2t\int_{0}^{t}\norma{b(\theta_{s}^{i},\mu_{\Theta_{s}^{N}})- b(\ov{\theta}_{s}^{i},f_{s})}_{1}^{2}\de s\\
 &\leq 4C^{2}t\int_{0}^{t}\left(\W_{2}\left(\mu_{\Theta_{s}^{N}},\mu_{\ov{\Theta}_{s}^{N}}\right)+\norma{\theta_{s}^{i}-\ov{\theta}_{s}^{i}}_{1}\right)^{2}\de s\\
 &+ 4C^{2}t\int_{0}^{t}\W_{2}^{2}\left(\mu_{\ov{\Theta}_{s}^{N}},f_{s}\right)\de s
\end{align}
where in the first inequality we have used Jensen inequality, and in the second inequality we have used the Lipschitz assumption. Then by Gr\hol{o}nwall inequality, we have
\begin{align}
\norma{\theta_{t}^{i}-\ov{\theta}_{t}^{i}}_{1}^{2} \leq 8C^{2}T\exp\left(8C^{2}T\right)\int_{0}^{t}\left(\W_{2}^{2}\left(\mu_{\Theta_{s}^{N}},\mu_{\ov{\Theta}_{s}^{N}}\right)+\W_{2}^{2}\left(\mu_{\ov{\Theta}_{s}^{N}},f_{s}\right)\right)\de s,
\end{align}
where $C$ is the Lipschitz constant in \eqref{Lipschitz}. Thus by summing overall $i$, we get
\begin{align}
\W_{2}^{2}\left(\mu_{\Theta_{t}^{N}},\mu_{\ov{\Theta}_{t}^{N}}\right)\leq 8C^{2}T\exp\left(8C^{2}T\right)\exp\left(8C^{2}T^{2}\exp\left(8C^{2}T\right)\right)\int_{0}^{t}\W_{2}^{2}\left(\mu_{\ov{\Theta}_{s}^{N}},f_{s}\right)\de s.
\end{align}
Then by \cite[Theorem 1]{fournier2013}, there exists a positive constant $C_{1}\coloneqq C_{1}(b,T)>0$ only depending on the Lipschtz constant of $b$ and $T$ such that
\begin{align}\label{newbou}
\begin{aligned}
\E\W_{2}^{2}\left(\mu_{\Theta_{t}^{N}},\mu_{\ov{\Theta}_{t}^{N}}\right)&\leq 8C^{2}T\exp\left(8C^{2}T\right)\exp\left(8C^{2}T^{2}\exp\left(8C^{2}T\right)\right)\int_{0}^{t}\E\W_{2}^{2}\left(\mu_{\ov{\Theta}_{s}^{N}},f_{s}\right)\de s\\
&\leq 8C^{2}T^{2}\exp\left(8C^{2}T\right)\exp\left(8C^{2}T^{2}\exp\left(8C^{2}T\right)\right)C_{1}\alpha_{d}(N).
\end{aligned}
\end{align}
Therefore,

\begin{align}
\frac{1}{N}\sum_{i=1}^{N}\E\sup_{0\leq t\leq T}\norma{\theta_{t}^{i}-\ov{\theta}_{t}^{i}}_{1}^{2}\leq  8C^{2}T^{2}\exp\left(8C^{2}T\right)(C_{2}(T)+1)\alpha_{d}(N)
\end{align}
where 
\begin{align}
C_{2}(T)\coloneqq 8C^{2}T^{2}\exp\left(8C^{2}T\right)\exp\left(8C^{2}T^{2}\exp\left(8C^{2}T\right)\right)C_{1}. 
\end{align}
\end{proof}

\begin{lem}\label{convergencelem}
Let $d>4$. Then 

\begin{align}\label{constantc1}
\E\W_{2}^{2}\left(\mu_{\Theta_{t}^{N}},\mu_{t}\right)\leq C_{1}\alpha_{d}(N),
\end{align}
where $C_{1}\coloneqq C_{1}(b,T)$ is a positive constant only depending on the Lipschitz constant of $b$, and $T$.
\end{lem}
\begin{proof}
Note that by the triangle inequality
\begin{align}
\E\W_{2}^{2}\left(\mu_{\Theta_{t}^{N}},\mu_{t}\right) \leq \E\W_{2}^{2}\left(\mu_{\Theta_{t}^{N}},\mu_{\ov{\Theta}_{t}^{N}}\right)+  \E\W_{2}^{2}\left(\mu_{\ov{\Theta}_{t}^{N}},\mu_{t}\right).
\end{align}
Then, by applying \eqref{newbou}, and \cite[Theorem 1]{fournier2013} we get the desired bound.
\end{proof}

\section{Mean field limit for general objective functions}\label{sec:generalmean}
In what follows, we consider the following objective function given by a uniform mixture of identical experts:

 \begin{align}\label{mixMoE}
    \begin{aligned}
        F(\Theta,x)&\coloneqq\frac{1}{N}\sum_{i=1}^{N}f(\theta^{i},x), \hskip 0,1cm \Theta\coloneqq (\theta^{1},\ldots, \theta^{N})\in (\mathbb{T}^{d})^{N},
        \end{aligned}
    \end{align}
    where $\mathbb{T}^{d}$ is the $d$-dimensional $2\pi$-periodic Torus, and $f:\mathbb{T}^{d}\times \mathcal{X}\rightarrow \R$ is a generic model function with suitable regularity properties. In what follows, as a cost function associated to the training set $\D$, we consider the mean squared error that is defined as

\begin{align}\label{costfunct1}
 \L(\Theta)\coloneqq \frac{1}{2}\sum_{j
=1}^{n}\left(F(\Theta,x_{j})-y_{j}\right)^{2}.  
\end{align}    
Next, we analyze the following gradient flow equation
\begin{align}\label{gradform1}
\frac{\de \Theta_{t}}{\de\,t}=-N\nabla_{\Theta}\L(\Theta_{t}), \hskip 0,2cm \Theta_{t}\in (\mathbb{T}^{d})^{N}.
\end{align}
We may write \eqref{gradform1} as follows:

\begin{align}\label{gradeq3}
\frac{\de \theta_{t}^{i}}{\de t}&=\sum_{j=1}^{n}\nabla_{\theta}f(\theta_{t}^{i},x_{j})\left(y_{j}-F(\Theta_{t},x_{j})\right), \hskip 0,2cm i=1,\ldots, N.
\end{align}

Let us define the {\em empirical probability measure} as follows

\begin{align}\label{eq:empiricalmeasure}
\mu_{\Theta_{t}^{N}}\coloneqq \frac{1}{N}\sum_{i=1}^{N}\delta_{\theta_{t}^{i}}.
\end{align}

Notice that by the definition of $F(\Theta_{t},x)$, it can be written as 

\begin{align}
F(\Theta_{t},x)= \frac{1}{N}\sum_{i=1}^{N}f(\theta_{t}^{i},x)=\E_{\ov{\alpha}\sim \mu_{\Theta_{t}^{N}}}\left[f(\ov{\alpha},x)\right].
\end{align}
Then, we may write \eqref{gradeq3} as 

\begin{align}\label{gradeq4}
\frac{\de \theta_{t}^{i}}{\de t}&=\sum_{j=1}^{n}\nabla_{\theta}f(\theta_{t}^{i},x_{j})\left(y_{j}-\E_{\ov{\alpha}\sim \mu_{\Theta_{t}^{N}}}\left[f(\ov{\alpha},x_{j})\right]\right), \hskip 0,2cm i=1,\ldots, N.
\end{align}
We let

\begin{align}\label{ourfield}
b(\theta_{t}^{i},\mu_{\Theta_{t}^{N}})\coloneqq\sum_{j=1}^{n}\nabla_{\theta}f(\theta_{t}^{i},x_{j})\left(y_{j}-\E_{\ov{\alpha}\sim \mu_{\Theta_{t}^{N}}}\left[f(\ov{\alpha},x_{j})\right]\right), \hskip 0,2cm i=1,\ldots, N,
\end{align}
and thus \eqref{gradeq4} can be given in the form

\begin{align}\label{gradeq5}
\frac{\de \theta_{t}^{i}}{\de t}=b(\theta_{t}^{i},\mu_{\Theta_{t}^{N}}), \hskip 0,2cm i=1,\ldots, N.
\end{align}
Next, we consider the following hypotheses about the model function $f$.

\begin{ass}\label{asummpt1}
We suppose that $f(\cdot,x)$ is continuous and $2\pi$-periodic in each component of $\theta$, that $\vert f(\cdot,x)\vert\leq 1$ for any $x\in\mathcal{X}$, and that the following bound holds true: there exists a positive constant $\alpha>0$ such that for any $i\in \{1,\ldots, d\}$,  and for any $x\in \mathcal{X}$
\begin{align}\label{A1eq1}
\vert\partial_{\theta_{i}}f(\theta,x)\vert\leq \alpha   
\end{align}
Furthermore, we suppose that  there exists a positive constant $\beta$ such that for any $i,j\in \{1,\ldots, d\}$, and any $x\in \mathcal{X}$, 
\begin{align}\label{A1eq2}
\vert \partial_{\theta_{i}}\partial_{\theta_{j}}f(\theta,x)\vert\leq \beta.  
\end{align}
\end{ass}
With this assumption, we are able to prove that $\nabla_{\theta}f(\theta,x)$, and $f(\theta)$ are Lipschitz with respect to the $\ell^{1}$-norm. This is the content of the following Proposition.

\begin{prop}\label{prop:general}
Let $(\mathbb{T}^{d},\norma{\cdot}_{1})$, and suppose that \autoref{asummpt1} holds true. Then 

\begin{align}\label{okaylipst1}
&\norma{\nabla_{\theta}f(\theta,x)-\nabla_{\theta}f(\theta',x)}_{1}\leq d\beta\norma{\theta-\theta'}_{1},\\\label{okaylipst2}
&\vert f(\theta,x)-f(\theta',x)\vert\leq \alpha\norma{\theta-\theta'}_{1}.
\end{align}
\end{prop}
\begin{proof}
Let us set  

\begin{align}
   dh_{x}(\theta):(\mathbb{T}^{d},\ell^{1})\rightarrow (\mathbb{T}^{d},\ell^{1}),\hskip 0,2cm v\mapsto \sum_{i=1}^{d}\partial_{\theta_{i}}\nabla_{\theta}f(\theta,x)v_{i}. 
\end{align}
Notice that

\begin{align}
    \norma{dh_{x}\theta}_{\ell^{1}\rightarrow \ell^{1}}&=\sup_{\norma{v}_{1}\leq 1}\norma{dh_{x}(\theta)v}_{1}\\
    &=\sup_{\norma{v}_{1}\leq 1}\norma{\sum_{i=1}^{d}\partial_{\theta_{i}}\nabla_{\theta}f(\theta,x)v_{i}}_{1}\\
    &=\sup_{\norma{v}_{1}\leq 1}\sum_{j=1}^{d}\sum_{i=1}^{d}\vert \partial_{\theta_{i}}\partial_{\theta_{j}}f(\theta,x)v_{i}\vert\\
    &\leq \sup_{\norma{v}_{1}\leq 1}\sum_{j=1}^{d}\sum_{i=1}^{d}\beta\vert v_{i}\vert
\end{align}
where in the last inequality we have used \eqref{A1eq2}. Hence

\begin{align}
    \norma{dh_{x}\theta}_{\ell^{1}\rightarrow \ell^{1}}\leq d\beta,
\end{align}
and thus

\begin{align}
\norma{\nabla_{\theta}f(\theta,x)-\nabla_{\theta}f(\theta',x)}_{1}&\leq \norma{dh_{x}(\theta)}_{\ell^{1}\rightarrow \ell^{1}}\norma{\theta-\theta'}_{1}\\
&\leq d\beta\norma{\theta-\theta'}_{1}.
\end{align}
On the other hand, we have 
\begin{align}
\vert f(\theta,x)-f(\theta',x) \vert\leq \max_{\theta\in \mathbb{T}^{d}}\norma{df(\theta,x)}_{{\rm op}}\norma{\theta-\theta'}_{1}.
\end{align} 
Let us estimate $\norma{df(\theta,x)}_{{\rm op}}$. By \eqref{A1eq1}, we get

\begin{align}
\norma{df(\theta,x)}_{{\rm op}}&=\sup_{\norma{v}_{1}\leq 1}\vert \nabla f(\theta,x)\cdot v\vert\\
&=\sup_{\norma{v}_{1}\leq 1}\sum_{i=1}^{d}\vert \partial_{\theta_{i}}f(\theta,x)v_{i}\vert\\
&=\sup_{\norma{v}_{1}\leq 1}\sum_{i=1}^{d} \alpha\vert v_{i}\vert\\
&\leq \alpha.
\end{align}
Therefore,

\begin{align}
 \vert f(\theta,x)-f(\theta',x) \vert\leq  \alpha\norma{\theta-\theta'}_{1}.
\end{align}
\end{proof}

\begin{thm}\label{thm:fieldlip}
Suppose that \autoref{asummpt1} holds true, and let $b:\mathbb{T}^{d}\times \M(\mathbb{T}^{d})\rightarrow \T^{d}$ be defined as in \eqref{ourfield}. Then 
\begin{align}
\norma{b(z_{1},\mu)-b(z_{2},\nu)}_{1} \leq C\left(\norma{z_{1}-z_{2}}_{1}+\W_{2}(\mu,\nu)\right),
\end{align}
where 
\begin{align}
C=\max\left\{d\beta n(A+1),\alpha^{2}dn\right\},\hskip 0,1cm A\coloneq \max_{1\leq j\leq n}\vert y_{j}\vert,
\end{align}
and $\W_{2}(\mu,\nu)$ is the $2$-Wasserstein distance for probability measures on $(\mathbb{T}^{d},\norma{\cdot}_{1})$.
\end{thm}

\begin{proof}[Proof of \autoref{thm:fieldlip}]
Notice that since the model function $f(\cdot,x)$ is $2\pi$-periodic, then $\nabla_{\theta}f(\cdot,x)$ too, and thus $b(\theta,\mu)$ is $2\pi$-periodic, so that the differential equation on $\T^{d}$ is well-defined. Let us take $z_{1},z_{2}\in \mathbb{T}^{d}$, and $\mu,\nu\in \M(\mathbb{T}^{d})$. Since $\vert f(\cdot,x)\vert\leq 1$, we get

\begin{align}
\norma{b(z_{1},\mu)-b(z_{2},\mu)}_{1} &\leq \sum_{k=1}^{d}\sum_{j=1}^{n}\vert \partial_{z_{k}}f(z_{1},x_{j})-\partial_{z_{k}}f(z_{2},x_{j})\vert \vert y_{j}-\E_{\ov{\alpha}\sim \mu}\left[f(\ov{\alpha},x_{j})\right]\vert\\
&\leq (1+A)\sum_{k=1}^{d}\sum_{j=1}^{n}\vert \partial_{z_{k}}f(z_{1},x_{j})-\partial_{z_{k}}f(z_{2},x_{j})\vert\\
&=(1+A)\sum_{j=1}^{n}\norma{\nabla_{\theta}f(z_{1},x_{j})-\nabla_{\theta}f(z_{2},x_{j})}_{1}\\
&\leq d\beta n(A+1)\norma{z_{1}-z_{2}}_{1},
\end{align}
where in the last inequality we have used \eqref{okaylipst1}. Further, by \eqref{okaylipst2}, $f(\theta,x)$ is Lipschitz with constant $\alpha$, then
\begin{align}
\norma{b(z_{1},\mu)-b(z_{1},\nu)}_{1} &\leq \sum_{k=1}^{d}\sum_{j=1}^{n}\vert \partial_{z_{k}}f(z_{1},x_{j})\vert \E_{\ov{\alpha}\sim \mu}\left[f(\ov{\alpha},x_{j})\right]- \E_{\ov{\alpha}\sim \nu}\left[f(\ov{\alpha},x_{j})\right]\vert\\
&\leq\alpha\sum_{k=1}^{d}\sum_{j=1}^{n}\vert \partial_{z_{k}}f(z_{1},x_{j})\vert\W_{1}(\mu,\nu)\\
&\leq \alpha^{2}dn\W_{1}(\mu,\nu),
\end{align}
where in the last two inequalities we have used \eqref{okaylipst2}. Therefore,

\begin{align}
\norma{b(z_{1},\mu)-b(z_{1},\nu)}_{1} &\leq \alpha^{2}dn\W_{1}(\mu,\nu)\\
&\leq \alpha^{2}dn\W_{2}(\mu,\nu)\\
\end{align}
where in the last inequality we have used \cite[Chapter 2, Formula 2.1]{panaretos2020invitation}.  Lastly, since
\begin{align}
C=\max\left\{d\beta n(A+1),\alpha^{2}dn\right\}
\end{align}
then

\begin{align}
\norma{b(z_{1},\mu)-b(z_{2},\nu)}_{1} \leq C\left(\norma{z_{1}-z_{2}}_{1}+\W_{2}(\mu,\nu)\right).
\end{align}
\end{proof}

\begin{thm}\label{ourmainthm}
Let us fix $T>0$, and $d>4$. Let us consider the system \eqref{gradeq5} with initial conditions $\Theta_{0}^{N}$ composed by independent and identically distributed random variables with values in $\mathbb{T}^{d}$. The following assertions hold true.
\begin{itemize}
\item[(I)] The system \eqref{gradeq5} has a unique strong solution.
\item[(II)] There exists a sequence of independent and identically distributed random variables $\ov{\Theta}_{t}^{N}$ for which propagation of chaos in the sense of Definition \ref{propchaos} holds true with $p=2$.
\item[(III)] The sequence $(\mu_{\Theta_{t}^{N}})$ weakly converges to $\mu_{t}\in \M(\T^{d})$  which is the unique solution of the continuity equation
 \begin{align}\label{FKP1}
\frac{\de \mu_{t}(\theta)}{\de t}=-\nabla_{\theta}\cdot\left(b(\theta,\mu_{t})\mu_{t}\right), \hskip 0,2cm \text{with initial condition $\mu_{0}$}.
\end{align}
Furthermore, the following bound holds true. There exists a positive constant $C_{1}\coloneqq C_{1}(b,T)$ only depending on the Lipschitz constant of $b$, and $T$ such that for each $t\in [0,T]$ one has
\begin{align}
\E\W_{2}^{2}\left(\mu_{\Theta_{t}^{N}},\mu_{t}\right)\leq C_{1}\alpha_{d}(N)
\end{align}
where 
\begin{align}
\alpha_{d}(N)\coloneqq N^{-\frac{2}{d}}+N^{-\frac{1}{2}}.
\end{align}
\end{itemize}
\end{thm}
\begin{proof}
By Corollary \ref{wellpodsyst}, we obtain the existence and uniqueness of a strong solution to \eqref{gradeq5}. Let us prove item (II). Notice that by \autoref{thmpropchaos}, there exists a sequence $\ov{\Theta}_{t}^{N}=(\ov{\theta}_{t}^{1}, \ldots, \ov{\theta}_{t}^{N})$ for which the propagation of chaos holds true. For the sake of simplicity, let us recall the procedure to find such a sequence. First notice that for any $i\ge1$, the sequence $(\theta_{t}^{i,N})_{t}$ defined through \eqref{gradeq5} weakly converges for $N\to\infty$ in $L^{2}(\Omega, C([0,T],\T^{d}))$ where $(\Omega,\P)$ is a common probability space for them. Indeed, let $N_{1}>N_{2}$, and consider $\Theta_{t}^{N_{1}}$ where the first $N_{2}$ components have the same initial condition as $\theta_{t}^{1,N_{2}},\ldots,\theta_{t}^{N_{2},N_{2}}$. We have that

\begin{align}
\E\left[\sup_{t\leq T}\norma{\theta_{t}^{1,N_{1}}- \theta_{t}^{1,N_{2}}}_{1}^{2}\right]\leq 2T\int_{0}^{T}\E\norma{b(\theta_{t}^{1,N_{1}},\mu_{\Theta_{t}^{N_{1}}})-b(\theta_{t}^{1,N_{2}},\mu_{\Theta_{t}^{N_{2}}})}_{1}^{2}\de t.
\end{align}
By following \cite[Theorem 3.1]{Chaintron_2022b}, we get

\begin{align}
\E\left[\sup_{t\leq T}\norma{\theta_{t}^{1,N_{1}}- \theta_{t}^{1,N_{2}}}_{1}^{2}\right]\leq C_{1}(b,T)\left(\frac{1}{N_{2}}-\frac{1}{N_{1}}\right) + C_{2}(b,T)\int_{0}^{T}\E\left[\norma{\theta_{t}^{1,N_{1}}- \theta_{t}^{1,N_{2}}}_{1}^{2}\right]\de t,
\end{align}
where $C_{1}(b,T), C_{2}(b,T)$ are positive constants only depending on $b$ and $T$. Thus by Gr\hol{o}nwall inequality, we get

\begin{align}\label{stimaL2}
\E\left[\sup_{t\leq T}\norma{\theta_{t}^{1,N_{1}}- \theta_{t}^{1,N_{2}}}_{1}^{2}\right] \leq  C_{1}(b,T)\left(\frac{1}{N_{2}}-\frac{1}{N_{1}}\right)\exp(C_{2}(b,T)T).
\end{align}
Then \eqref{stimaL2} implies that $(\theta_{t}^{1,N})$ is a Cauchy sequence in $L^{2}(\Omega, C([0,T],\T^{d}))$. Then  there exists a variable $\ov{\theta}_{t}^{1}$ which is the  limit of such a sequence in the space $L^{2}(\Omega, C([0,T],\T^{d}))$. Applying the same reasoning for any $k\in\N$, we find $\ov{\theta}_{t}^{k}$ as the limit of $(\theta_{t}^{k,N})$. Furthermore, by following \cite[Proof of Theorem 3.1, Step 3]{Chaintron_2022b}, we get

\begin{align}\label{seq}
\ov{\theta}_{t}^{k}=\theta_{0}^{k}+ \int_{0}^{t}b(\ov{\theta}_{s}^{k},\mu_{s})\de s
\end{align}
where $\mu_{t}={\mathrm Law}(\ov{\theta}_{t}^{k})$, and the all variables $(\ov{\theta}_{t}^{k})_{t}$, $k\in\N$, are independent. We can take this sequence as the one required in \autoref{propchaos}. Now,  let us now note that item (III) follows by \autoref{convergencelem}, and \autoref{wellpod}.
\end{proof}

\section{Experts given by quantum circuits}\label{sub:quantum}
In this section, we denote by $m\in \N$ the number of qubits of the quantum circuit implementing each expert. %Hence, the Hilbert space of the system is $\H=\left(\CC^{2}\right)^{\otimes m}$, and its dimension denoted as $\dim\,\H$ is $2^{m}$.  
Let $\CC^{2}$ be the Hilbert space of a single qubit. We consider an observable $\O$ on the Hilbert space $\H=\left(\CC^{2}\right)^{\otimes m}$ with $\|\O\|\le1$.
The model function of each expert is then
    \begin{align}\label{model1}
    f(\theta,x)&\coloneqq \bra{0^{m}}U^{\dag}(\theta,x)\O\,U(\theta,x)\ket{0^{m}},    
    \end{align}
where $U$ is a unitary operation given by the expression

\begin{align}\label{uformula}
U(\theta,x)\coloneqq V_{d}(x)e^{-\frac{i\theta_{d}\mathcal{G}_{d}}{2}}V_{d-1}(x)\cdots V_{1}e^{-\frac{i\theta_{1}\mathcal{G}_{1}}{2}}V_{0}(x)    
\end{align}
where $V_{j}\in \mathcal{L}(\H)$, $j=0,\ldots,d$ are unitary operations, and  $\mathcal{G}_{i}$ hermitian operators such that $\norma{\mathcal{G}_{i}}\leq 1$ for all $i=1,\ldots,d$. Now, we consider the quantum neural network defined as a classical mixture of $N$ identical quantum experts with independent parameters.
The corresponding model function is given by \eqref{mixMoE} where $f$ as in \eqref{model1}, that is
\begin{equation}\label{eq:quantumMoE}
    F(\Theta,x)=\frac{1}{N}\sum_{i=1}^N \bra{0^{m}}U^{\dag}(\theta^i,x)\O\,U(\theta^i,x)\ket{0^{m}}.  
\end{equation}  

\noindent
Let us remark that \cite{girardi2024, hernandez2024} consider the training of quantum neural networks of the form
\begin{align}\label{girher}
    F(\Theta,x)=\frac{1}{N(M)}\sum_{i=1}^M\bra{0^{M}}\mathcal{U}^{\dag}(\Theta,x)\O_i\,\mathcal{U}(\Theta,x)\ket{0^{M}},    
    \end{align}
where $M$ is the number of qubits, $\mathcal{U}(\Theta,x)$ is a parametric $M$-qubit unitary operator, each $\O_i$ is  a  single-qubit observable and $N(M)$ is a suitable normalizing constant depending on the number of qubits $M$ chosen such that the covariance of $F(\Theta,x)$ at initialization has a finite nonzero limit for $M\to\infty$. They show that under suitable hypotheses, the displacement of each component of $\Theta_{t}$ (trained by gradient flow) from the corresponding initial value is bounded by a quantity that becomes arbitrarily small as $M\rightarrow +\infty$ uniformly in time. Such a regime is called lazy training. 
In the present work, the total number of qubits of the network is $Nm$. Therefore, the limit of infinite width $M\rightarrow +\infty$ considered in \cite{girardi2024,hernandez2024} is here replaced with the limit $N\to\infty$. In \cite{girardi2024, hernandez2024}, the variance at initialization of the model function is constant, while the function \eqref{eq:quantumMoE} is uniformly bounded in $N$ and its variance at initialization scales as $1/N$. Therefore, our quantum network is not in the lazy regime and can have effective representation learning. We stress that none of our results depends on the number of layers. Since we consider all the experts given by the same fixed quantum circuit, we do not consider the limit of infinite depth.

In what follows, we consider the case of a model function $f$ generated by a quantum circuit. Next, we then aim to show that under suitable assumptions on the quantum circuit, we can provide an explicit formula for the constant $\alpha$, and $\beta$ founded in Proposition \ref{prop:general}.
\begin{lem}\label{mainlem}
Let $f$ be defined according to \eqref{model1}. Then \eqref{model1} satisfies \autoref{asummpt1}, and the following holds true.
\begin{align}\label{gradflips}
&\norma{\nabla_{\theta}f(\theta,x)-\nabla_{\theta}f(\theta',x)}_{1}\leq d\norma{\theta-\theta'}_{1},\\\label{flips}
&\vert f(\theta,x)-f(\theta',x)\vert\leq \norma{\theta-\theta'}_{1},
\end{align}
so that 
\begin{align}
\alpha=\beta=1.
\end{align}
\end{lem}
\begin{proof}
Let us first prove that for each $\theta=(\theta_{1},\ldots,\theta_{d})$, one gets that
\begin{align}\label{girher1}
 &\vert \partial_{\theta_{k}}f(\theta,x) \vert\leq 1,\\\label{girher2}
 &\vert \partial_{\theta_{j}}\partial_{\theta_{k}}f(\theta,x)\vert\leq 1
\end{align}
for all $j,k=1,\ldots,d$. Next, let us set $W_{j}(\theta_{j})\coloneqq e^{-\frac{i\theta_{j}\mathcal{G}_{j}}{2}}$ for all $j=1,\ldots, d$, and $V_{k}\coloneqq V_{k}(x)$, for all $k=0,\ldots, d$. By the definition of $U$ as in \eqref{uformula}, we get that
\begin{align}
\partial_{\theta_{k}}f(\theta,x)&=\bra{0^{m}}V_{0}^{\dag}(x)W_{1}^{\dag}(\theta_{1})V_{1}^{\dag}\cdots\left(\partial_{\theta_{k}}W_{k}(\theta_{k})\right)^{\dag}V_{k}^{\dag}\cdots W_{d}^{\dag}(\theta_{d})V_{d}^{\dag}\O V_{d}W_{d}(\theta_{d})\cdots V_{1}W_{1}(\theta_{1})V_{0}\ket{0^{m}}\\
&+\bra{0^{m}}V_{0}^{\dag}(x)W_{1}^{\dag}(\theta_{1})V_{1}^{\dag}\cdots W_{d}^{\dag}(\theta_{d})V_{d}^{\dag}\O V_{d}W_{d}(\theta_{d})\cdots V_{k}\left(\partial_{\theta_{k}}W_{k}(\theta_{k})\right) \cdots V_{1}W_{1}(\theta_{1})V_{0}\ket{0^{m}}\\
&=\frac{i}{2}\bra{0^{m}}V_{0}^{\dag}(x)W_{1}^{\dag}(\theta_{1})V_{1}^{\dag}\cdots \mathcal{G}_{k}W_{k}^{\dag}(\theta_{k})V_{k}^{\dag}\cdots W_{d}^{\dag}(\theta_{d})V_{d}^{\dag}\O V_{d}W_{d}(\theta_{d})\cdots V_{1}W_{1}(\theta_{1})V_{0}\ket{0^{m}}\\
&-\frac{i}{2}\bra{0^{m}}V_{0}^{\dag}(x)W_{1}^{\dag}(\theta_{1})V_{1}^{\dag}\cdots W_{d}^{\dag}(\theta_{d})V_{d}^{\dag}\O V_{d}W_{d}(\theta_{d})\cdots V_{k}W_{k}(\theta_{k})\mathcal{G}_{k}\cdots V_{1}W_{1}(\theta_{1})V_{0}\ket{0^{m}}.
\end{align}
Since $\norma{V_{k}}\leq 1$, and $\norma{W_{j}}\leq 1$, we have that

\begin{align}
\vert\partial_{\theta_{k}}f(\theta,x)\vert\leq \norma{\mathcal{G}_{k}}\norma{\O}\leq 1. 
\end{align}
Let us now compute $\partial_{\theta_{j}}\partial_{\theta_{k}}f(\theta,x)$. Notice that

\begin{align}
\partial_{\theta_{j}}\partial_{\theta_{k}}f(\theta,x)&= \frac{i^{2}}{4}\bra{0^{m}}V_{0}^{\dag}(x)W_{1}^{\dag}(\theta_{1})V_{1}^{\dag}\cdots \mathcal{G}_{j}W_{j}^{\dag}(\theta_{j})V_{j}^{\dag}\cdots \\
&\cdots \mathcal{G}_{k}W_{k}^{\dag}(\theta_{k})V_{k}^{\dag}\cdots W_{d}^{\dag}(\theta_{d})V_{d}^{\dag}\O V_{d}W_{d}(\theta_{d})\cdots V_{1}W_{1}(\theta_{1})V_{0}\ket{0^{m}}\\
&+\frac{i(-i)}{4}\bra{0^{m}}V_{0}^{\dag}(x)W_{1}^{\dag}(\theta_{1})V_{1}^{\dag}\cdots\mathcal{G}_{k}W_{k}^{\dag}(\theta_{k})V_{k}^{\dag}\cdots W_{d}^{\dag}(\theta_{d})V_{d}^{\dag}\O V_{d}W_{d}(\theta_{d})\cdots\\
&\cdots V_{j}W_{j}(\theta_{j})\mathcal{G}_{j}\cdots V_{1}W_{1}(\theta_{1})V_{0}\ket{0^{m}}\\
&+ h.c..
\end{align}
Since $\norma{V_{k}}\leq 1$, and $\norma{W_{j}}\leq 1$, we then get that

\begin{align}
\vert \partial_{\theta_{j}}\partial_{\theta_{k}}f(\theta,x)\vert \leq \norma{\mathcal{G}_{j}}\norma{\mathcal{G}_{k}}\norma{\O}\leq 1. 
\end{align}
Furthermore, notice that by construction $f$ is analytic, and satisfies $\vert f(\theta,x)\vert\leq 1$ for all $\theta,x$. Let us set  

\begin{align}
   dh_{x}(\theta):(\T^{d},\ell^{1})\rightarrow (\T^{d},\ell^{1}),\hskip 0,2cm v\mapsto \sum_{i=1}^{d}\partial_{\theta_{i}}\nabla_{\theta}f(\theta,x)v_{i}. 
\end{align}
Notice that

\begin{align}
    \norma{dh_{x}\theta}_{\ell^{1}\rightarrow \ell^{1}}&=\sup_{\norma{v}_{1}\leq 1}\norma{dh_{x}(\theta)v}_{1}\\
    &=\sup_{\norma{v}_{1}\leq 1}\norma{\sum_{i=1}^{d}\partial_{\theta_{i}}\nabla_{\theta}f(\theta,x)v_{i}}_{1}\\
    &=\sup_{\norma{v}_{1}\leq 1}\sum_{j=1}^{d}\sum_{i=1}^{d}\vert \partial_{\theta_{i}}\partial_{\theta_{j}}f(\theta,x)v_{i}\vert\\
    &\leq \sup_{\norma{v}_{1}\leq 1}\sum_{j=1}^{d}\sum_{i=1}^{d}\vert v_{i}\vert
\end{align}
where in the last inequality we have used \eqref{girher2}. Therefore
\begin{align}
\norma{dh_{x}\theta}_{\ell^{1}\rightarrow \ell^{1}}\leq d.
    \end{align}
Hence,

\begin{align}
\norma{\nabla_{\theta}f(\theta,x)-\nabla_{\theta}f(\theta',x)}_{1}&\leq \norma{dh_{x}(\theta)}_{\ell^{1}\rightarrow \ell^{1}}\norma{\theta-\theta'}_{1}\\
&\leq d\norma{\theta-\theta'}_{1}.
\end{align}
On the other hand, we have 
\begin{align}
\vert f(\theta,x)-f(\theta',x) \vert\leq \max_{\theta\in \mathbb{T}^{d}}\norma{df(\theta,x)}_{{\rm op}}\norma{\theta-\theta'}_{1}.
\end{align} 
Let us estimate $\norma{df(\theta,x)}_{{\rm op}}$. By \eqref{girher1}

\begin{align}
\norma{df(\theta,x)}_{{\rm op}}&=\sup_{\norma{v}_{1}\leq 1}\vert \nabla f(\theta,x)\cdot v\vert\\
&=\sup_{\norma{v}_{1}\leq 1}\sum_{i=1}^{d}\vert \partial_{\theta_{i}}f(\theta,x)v_{i}\vert\\
&=\sup_{\norma{v}_{1}\leq 1}\sum_{i=1}^{d} \vert v_{i}\vert\\
&\leq 1.
\end{align}
Therefore,

\begin{align}
 \vert f(\theta,x)-f(\theta',x) \vert\leq \norma{\theta-\theta'}_{1}.
\end{align}
\end{proof}

\begin{thm}[mean-field limit of quantum neural networks]\label{mfieldq}
Let us fix $T>0$, and $d>4$. Let us consider the system \eqref{gradeq5} with initial conditions $\Theta_{0}^{N}$ composed by independent and identically distributed random variables with values in $\mathbb{T}^{d}$. Furthermore, let $f$ be  generated by a quantum circuit according to \eqref{model1}. Then the same assertions as in \autoref{ourmainthm} hold true.
\end{thm}
\begin{proof}
Let us notice that by \autoref{mainlem}, we are in position to apply \autoref{wellpodsyst}, and \autoref{thmpropchaos}. Furthermore, since $\alpha=\beta=1$, the involved Lipschitz constant in \autoref{thm:fieldlip} is then
\begin{align}
C=dn(A+1).
\end{align}
So that our conclusion follows by applying the same reasoning as in \autoref{ourmainthm}. 
\end{proof}

\section{Conclusions}\label{sec:concl}
In this paper we have studied the training of mixtures of identical experts with particular focus on the case where each expert is given by a quantum neural network. We have applied the mean-field limit and the propagation of chaos to model the time evolution of the model parameters in the training by gradient flow. Within such a framework, the parameters of a single expert are considered as the spatial coordinates of a particle, and the overall training is described by the dynamics of a system of particles induced by a vector field in the sense of \eqref{gradeq4} and \eqref{ourfield}. As a consequence, a general equation of continuity can be defined. In \autoref{ourmainthm}, we show that the empirical measure $\mu_{\Theta_t^N}$ associated to the trained parameters, defined in \eqref{eq:empiricalmeasure}, weakly converges to the solution $\mu_t$ of the continuity equation \eqref{FKP1}. More precisely, we provide a quantitative bound over the Wasserstein of order 2 between $\mu_{\Theta_t^N}$ and $\mu_t$ which tends to zero as the number of experts in the mixture goes to infinity. Then, we have focused on the case where each expert in the mixture is a parametric quantum circuit, specializing the convergence result (\autoref{mfieldq}). Differently from \cite{girardi2024,hernandez2024}, the training of the quantum neural network considered here does not happen in the lazy regime, enabling representation learning. Our techniques do not allow us to study the joint limit of infinite depth and width.
%Our main result is a quantitative proof of the convergence to a probability density determined by a nonlinear continuity equation of the empirical measure generated by a sequence of trained parameters by gradient flow. In \autoref{ourmainthm}, we have provided an upper bound to the Wasserstein distance of order $2$ between such probability density  and such empirical measures.

Our result open the way to several possible research directions:
\begin{itemize}
    
    \item  Finding a better rate of convergence for the sequence of empirical measures generated by the vector $\Theta_{t}^{N}$, where the rate of convergence is polynomial in the number of parameters of each expert rather than exponential, as found in the present work.
    
    \item  Determining time-uniform upper bounds for the Wasserstein distance of order $2$ between the empirical distribution of the parameters and the limit probability measure. Such bounds would prove that the mean-field approximation holds even for $t\to\infty$, \emph{i.e.}, at the end of the training when the generated function reproduces perfectly the training examples.

    \item Extending our results to the setting where the number of parameters of each expert grows with $N$. This setting would allow us to consider the joint limit of infinite depth and width, and to scale the complexity of each expert with $N$. In this setting, the convergence of the empirical distribution of the parameters is ill-defined, and the probability distribution of the generated function would have to be considered.
\end{itemize}

\section*{Acknowledgements}
GDP has been supported by the HPC Italian National Centre for HPC, Big Data and Quantum Computing -- Proposal code CN00000013 -- CUP J33C22001170001 and by the Italian Extended Partnership PE01 -- FAIR Future Artificial Intelligence Research -- Proposal code PE00000013 -- CUP J33C22002830006 under the MUR National Recovery and Resilience Plan funded by the European Union -- NextGenerationEU.
Funded by the European Union -- NextGenerationEU under the National Recovery and Resilience Plan (PNRR) -- Mission 4 Education and research -- Component 2 From research to business -- Investment 1.1 Notice Prin 2022 -- DD N. 104 del 2/2/2022, from title ``understanding the LEarning process of QUantum Neural networks (LeQun)'', proposal code 2022WHZ5XH -- CUP J53D23003890006.
DP has been supported by project SERICS (PE00000014) under the MUR National Recovery and Resilience Plan funded by the European Union -- NextGenerationEU.
GDP and DP are members of the ``Gruppo Nazionale per la Fisica Matematica (GNFM)'' of the ``Istituto Nazionale di Alta Matematica ``Francesco Severi'' (INdAM)''. AMH has been supported by project PRIN 2022 
``understanding the LEarning process of QUantum Neural networks (LeQun)'', proposal code 2022WHZ5XH -- CUP J53D23003890006. The author AMH is a member of the ``Gruppo Nazionale per l'Analisi Matematica, la Probabilità e le loro Applicazioni (GNAMPA)'' of the ``Istituto Nazionale di Alta Matematica ``Francesco Severi'' (INdAM)''.

\bibliographystyle{siam}
\bibliography{bibliography_1}
 \end{document}